\documentclass[sigconf]{acmart}

\AtBeginDocument{%
  }

\setcopyright{acmlicensed}
\copyrightyear{2018}
\acmYear{2018}
\acmDOI{XXXXXXX.XXXXXXX}
\acmConference[Conference acronym 'XX]{Make sure to enter the correct
  conference title from your rights confirmation email}{June 03--05,
  2018}{Woodstock, NY}
\acmISBN{978-1-4503-XXXX-X/2018/06}

\usepackage{booktabs}
\usepackage{verbatim}
\usepackage[thinc]{esdiff}
\usepackage{pgfplots}
\usepgfplotslibrary{patchplots}
\usepgfplotslibrary{groupplots}
\usetikzlibrary{patterns}

\usepackage[ruled]{algorithm}
\usepackage{algpseudocode}
\usepackage{algorithmicx}
\usepackage{xcolor}
\colorlet{kw}{blue}
\definecolor{comments}{rgb}{0,0.5,0}

\newcommand{\algkeyword}[1]{\textcolor{kw}{\textbf{#1}}}
\algrenewcommand\algorithmicfunction{\algkeyword{function}}
\algrenewcommand\algorithmicwhile{\algkeyword{while}}
\algrenewcommand\algorithmicfor{\algkeyword{for}}
\algrenewcommand\algorithmicdo{\algkeyword{do}}
\algrenewcommand\algorithmicif{\algkeyword{if}}
\algrenewcommand\algorithmicthen{\algkeyword{then}}
\algrenewcommand\algorithmicelse{\algkeyword{else}}
\algrenewcommand\algorithmicend{\algkeyword{end}}
\algrenewcommand\algorithmicreturn{\algkeyword{return}}

\begin{document}

\title{Implementing Semantic Join Operators Efficiently}

\author{Immanuel Trummer}
\affiliation{%
  \institution{Cornell University}
  \city{Ithaca}
  \state{USA}
}
\email{itrummer@cornell.edu}


\begin{abstract}
Semantic query processing engines often support semantic joins, enabling users to match rows that satisfy conditions specified in natural language. Such join conditions can be evaluated using large language models (LLMs) that solve novel tasks without task-specific training.

Currently, many semantic query processing engines implement semantic joins via nested loops, invoking the LLM to evaluate the join condition on row pairs. Instead, this paper proposes a novel algorithm, inspired by the block nested loops join operator implementation in traditional database systems. The proposed algorithm integrates batches of rows from both input tables into a single prompt. The goal of the LLM invocation is to identify all matching row pairs in the current input. The paper introduces formulas that can be used to optimize the size of the row batches, taking into account constraints on the size of the LLM context window (limiting both input and output size). An adaptive variant of the proposed algorithm refers to cases in which the size of the output is difficult to estimate. A formal analysis of asymptotic processing costs, as well as empirical results, demonstrates that the proposed approach reduces costs significantly and performs well compared to join implementations used by recent semantic query processing engines.
\end{abstract}

\maketitle

\section{Introduction}
\label{sec:introduction}

Several recent systems~\cite{Trummer2022, Narayan2022, Arora, Kayali2023, Saeed2023, Trummer2022b, Chen2023, Thorne2021} expand SQL by introducing semantic operators. Those operators, including, for instance, semantic filters and semantic sort operators, are configured via natural language instructions and evaluated by large language models (LLMs). Compared to traditional relational operators, the per-byte processing overheads of such operators are typically higher by many orders of magnitude. This means, in the context of semantic queries, processing overheads are typically dominated by overheads due to semantic operators. This makes it crucial to make those operators as efficient as possible.

This paper focuses on a semantic version of a classical relational operator: the relational join. Semantic joins, as defined by systems like LOTUS~\cite{Patel2025}, enable users to define the join condition in natural language. Note that this paper does not explicitly focus on equality joins. Instead, it focuses on general theta-joins~\cite{Mishra1992} with natural language predicates. Such join operators are useful in the following example scenarios.



\begin{example}
    To investigate a large corporation\footnote{This example is motivated by the investigation of the Enron corporation. In the context of this investigation, prosecutors had access to and analyzed some of over 500,000 emails from more than 150 employees.}, prosecutors plan to analyze a large collection of emails. The goal is to compare emails to statements made by executives of that company. Of particular interest are instances where an email contradicts statements made by the defendants. This can be modeled as a join between two data sets, one containing statements and the other one containing emails. The join predicate can be formulated in natural language as ``pairs of documents that contradict each other.''
\end{example}

\begin{example}
    A Website enables users to enter ads for used goods as free text, as well as free text descriptions of items they are searching for. Free text descriptions may contain detailed descriptions of various aspects of the items offered (or desired). For instance, a user looking for a new table may include constraints and preferences with regard to material, color, size, or state (e.g., ``no coffee stains''). The Website wants to introduce a feature that supports users in matching ads to requests. This can be formulated as a join between two tables containing ads and searches. The join predicate can be expressed in natural language as ``pairs of ads matching requests.''
\end{example}

Both tasks require natural language understanding as well as, potentially, some common-sense knowledge. E.g., determining that ``I met Chris in Houston in the afternoon'' contradicts ``I saw Chris in Berlin at 1 PM'' (assuming both statements refer to the same day and person) requires commonsense knowledge in terms of the minimal duration of a flight between the two locations. Fortunately, language models like OpenAI's GPT models combine both capabilities and can be used, in principle, to evaluate the join predicates in the two aforementioned scenarios (as well as many others).

Perhaps the most natural way to use language models in the aforementioned scenarios is to compare pairs of entries from the two input collections. Iterating over all pairs of items from the two inputs, the language models can be tasked to compare two specific items with regard to the user-specified join condition. Essentially, this corresponds to a tuple nested loop join with the language model as predicate evaluation function (a corresponding approach is described in more detail in Section~\ref{sec:TupleJoin}). To comply with database terminology, we will from now on refer to the two input collections as ``tables'' and to their elements as ``tuples'' (even if the items do not actually have to correspond to relational tuples).


The problem with the aforementioned approach is efficiency. Using LLMs is expensive. Providers such as OpenAI charge per the amount of text (measured in ``tokens'', the atomic unit at which language models represent text) read and generated. Pairwise comparisons between a moderate number of 10,000 ads and searches with 100 tokens each (which corresponds to about a paragraph of text) using GPT-4 would cost about 600,000 dollars according to current rates\footnote{At the time of writing, OpenAI charges 3 cents per 1,000 tokens read when using GPT-4.}. To make such approaches practical, it is crucial to reduce processing fees.

This paper shows that techniques from traditional database join algorithms~\cite{Mishra1992, Ramakrishnan2002} can be adapted to improve efficiency for theta-joins, executed by language models, by many orders of magnitude. The core insight behind the algorithms presented in the following sections is the fact that language models can be used much more flexibly than pure predicate evaluation functions. More precisely, it is possible to task language models with finding matching pairs directly, within a set of tuples provided as part of the input. All it takes is a corresponding instruction in natural language, provided as part of the input as well. However, this does not mean that it is possible to use language models to perform the entire join operation in a single step. Language models come with strict constraints on the number of tokens that can be processed (i.e., read or generated) during a single invocation of the model. To scale the approach to larger data sets, it is necessary to decompose corresponding joins into multiple invocations of the LLM, each invocation only referring to a limited subset of the input data.

The approach described above resembles a block nested loops join. Similar to the classical join algorithm, the goal is to deal with limits on the data size that can be stored in the higher levels of the memory hierarchy, i.e., closer to the processor. In this scenario, the language model is the ``processor'' and size limits are imposed by the per-invocation token limit, intrinsic to the language model used. Despite those similarities, language models come with a very specific set of constraints, rendering a straightforward adaption of the traditional join operators inefficient.

For instance, for classical block nested loops joins, it is assumed that an output buffer of minimal size suffices. This assumption is justified if the output buffer can be repeatedly flushed to disk and re-filled during join processing, without losing the content of input buffers. However, in the case of language models, reaching the token limit when generating output tokens means that the request terminates. In that case, all input sent to the language model is lost and must be resent in a follow-up request (which incurs costs for reading input tokens again). This means, rather than generating join results gradually, we need to ensure that, for each invocation of the language model, the complete result fits within the token limit, determined by the token limit of the model, minus the number of tokens used up for the input data (and join task description).

The fact that a hard bound limits the combined size of the task description, (partial) join inputs, and join outputs is unique to the context of language models and requires careful planning. Sending too much input data within a single request is risky as the complete join result may exceed the token limit, possibly requiring redoing the corresponding task. On the other hand, sending less input than possible is inefficient. It means that the language model is invoked more often than necessary. In the worst case, this approach reduces, essentially, to the tuple join, which uses language models to analyze single pairs of input tuples in each invocation.

The following sections develop, first of all, a custom cost model, calculating join processing costs as a function of input data properties (e.g., the number and average token size of the input tuples), language model properties (e.g., the cost per token read or generated and the maximal number of tokens processed per invocation), as well as of the batch sizes chosen for the two input tables (i.e., the number of tuples from each of the two input tables, sent per model invocation). The cost model focuses on fees paid for using the language model, typically the dominant cost factor when using LLMs such as GPT-4. Whereas data and model properties cannot be influenced, the number of tuples sent per model invocation can be chosen. Therefore, along with the cost model, the following sections derive formulas for calculating the optimal batch size for both input tables, given values for all relevant parameters describing data and model properties.

The cost model, and therefore the formulas for calculating optimal batch sizes, rely crucially on the selectivity of the join predicate. This selectivity determines how many join output tuples are generated, in expectation, per model invocation. Therefore, the selectivity determines how much ``space,'' i.e., how many tokens, need to be reserved for generating output as opposed to storing input tuples. The lower the selectivity, the fewer tokens need to be reserved for writing output. This means we can send more input tuples in each invocation of the language model, reducing the number of LLM invocations required to generate a complete join result (and therefore, as we will see in the following sections, the expected join costs).

As join predicates are formulated in natural language, it is not possible to apply standard methods to estimate their selectivity (e.g., based on histograms or other data statistics). However, it turns out that knowing the precise selectivity in advance is ultimately not necessary. This paper presents an adaptive join algorithm that automatically adapts join selectivity estimates, along with the associated choices for batch sizes. Starting from an optimistic selectivity estimate, i.e., an estimate that is possibly much smaller than the actual selectivity, the adaptive join algorithm starts by sending batches of tuples that may be too large to be processed in a single model invocation (since the amount of output generated exceeds the token limit). By a suitable design of the task instructions for the LLM, cases in which an incomplete result is generated due to the token limit can be recognized (we will use the term ``overflow'' in such cases). In the case of an overflow, the adaptive join algorithm updates the selectivity estimate by increasing it by a constant factor. Eventually, the selectivity reaches an estimate that is equal to or higher than the actual selectivity. This means that sending tuples does not result in an overflow anymore.

While it is clear that the adaptive join algorithm will eventually find a selectivity estimate that avoids overflows, it is not clear, a-priori, that this approach results in interesting performance properties. However, formal analysis shows that the adaptive join algorithm reaches near-optimal join processing costs under moderately simplifying assumptions.


The experiments, using OpenAI's GPT-4 model, demonstrate that batching tuples in join prompts leads to a dramatic reduction in semantic join costs. Specifically, the proposed join algorithms reduce processing overheads significantly compared to join algorithms used in LOTUS~\cite{Patel2025}, a recently proposed semantic query processing engine. Comparing different join implementations proposed in this paper, it turns out that the block join works best if the selectivity of the join predicate is known. On the other hand, the adaptive version achieves nearly the same performance without requiring a selectivity value beforehand. A simple approach exploiting embedding vectors to match row pairs during the join works best in scenarios where the join condition is semantically close to an equality join. In scenarios where the goal is to match items that are complementary (e.g., matching contradicting statements), the result quality may, however, degrade severely. 



In summary, the original scientific contributions in this paper are the following:

\begin{itemize}
    \item The paper introduces multiple algorithms implementing semantic joins with arbitrary (i.e., not necessarily equality) join conditions, described in natural language.
    \item The paper analyzes the cost of the proposed algorithms in terms of token consumption, proposing formulas to tune these implementations for optimal performance.
    \item The paper presents experiments, evaluating the proposed algorithms in several scenarios, comparing to multiple baseline algorithms (some of which are currently used in semantic query processing engines).
\end{itemize}

The remainder of this paper is organized as follows. Section~\ref{sec:Model} introduces the problem model and related terminology. Section~\ref{sec:TupleJoin} describes a simple join algorithm that uses language models for pairwise tuple comparisons. Section~\ref{sec:BlockJoin} describes a join operator that exploits LLMs for finding matching pairs between tuple batches. Section~\ref{sec:Optimization} shows how to optimize batch sizes for that join operator if the selectivity of the join predicate is known. Section~\ref{sec:Adaptive} presents an adaptive join operator that automatically updates selectivity estimates while achieving near-optimal performance. Section~\ref{sec:experiments} reports on experiments, comparing all join operators in different scenarios and according to different metrics. Finally, Section~\ref{sec:Related} contrasts the work presented in this paper with prior work.
\section{Problem Model}
\label{sec:Model}

This paper addresses the following problem.

\begin{definition}[Semantic Join with Natural Language Predicates]
    Given two tables $R_1$ and $R_2$, together with a join predicate $j$, expressed as free text in natural language, find all pairs $R\subseteq R_1\times R_2$ that satisfy predicate $j$.
\end{definition}

Tuples may represent text documents or a textual representation of structured records. The aforementioned problem can be solved by LLMs.

\begin{definition}[Large Language Model]
    A large language model processes arbitrary tasks, described in natural language in the prompt (the input text sent to the model). Processing fees are proportional to the number of tokens (the atomic unit at which text is represented) read and generated (with possibly different cost factors for tokens read and generated). The sum of tokens read and generated per model invocation is upper-bounded by a model-specific constant.
\end{definition}

\begin{table}[t]
    \centering
    \caption{Symbols and their semantics.\label{tab:symbols}}
    \begin{tabular}{ll}
    \toprule[1pt]
    \textbf{Symbol} & \textbf{Semantics} \\
    \midrule[1pt]
        $r_i$ & Number of rows in table $i$\\
        $b_i$ & Number of rows per batch for table $i$ \\
        $b_i^*(\sigma)$ & Optimal batch size for table $i$\\
        $s_i$ & Token size per entry in table $i$ \\
        $\sigma$ & Selectivity of join condition \\
        $g$ & Relative cost of generating tokens \\
        $p$ & Size of task description with predicate \\
        $t$ & Token threshold per LLM invocation \\
        $c(b_1,b_2)$ & Total processing costs \\
        $c^*(b_1)$ & Cost for given $b_1$ and optimal choice of $b_2$ \\
        $o(e,\sigma)$ & Join cost with selectivity $\sigma$ when optimizing for $e$ \\
    \bottomrule[1pt]
    \end{tabular}
\end{table}

Table~\ref{tab:symbols} summarizes all symbols introduced in the next sections. 
\section{Tuple Nested Loops Join}
\label{sec:TupleJoin}

This section introduces a variant of the tuple nested loops join, as well as an associated cost model.

\subsection{Algorithm}
\label{sub:algorithm}

\begin{algorithm}[t]
\begin{algorithmic}[1]
    \State \Comment{Perform tuple join between relations $R_1$ and $R_2$,}
    \State \Comment{using join condition $j$.}
    \Function{BlockJoin}{$R_1,R_2,j$}
    \State \Comment{Initialize result set}
    \State $R\gets\emptyset$
    \State \Comment{Iterate over tuple pairs}
    \For{$t_1\in R_1$}
    \For{$t_2\in R_2$}
    \State \Comment{Create prompt for LLM}
    \State $P\gets$\Call{TuplePrompt}{$t_1,t_2,j$}
    \State \Comment{Ask LLM if join condition satisfied}
    \State $A\gets$\Call{InvokeLLM}{$P$}
    \State \Comment{Add result tuple if answer is positive}
    \If{A=``Yes''}
    \State $R\gets R\cup\{\langle t_1,t_2\rangle\}$
    \EndIf
    \EndFor
    \EndFor
    \State \Comment{Return join result}
    \State \Return{$R$}
    \EndFunction
\end{algorithmic}
    \caption{Tuple nested loops join algorithm for semantic joins, executed via large language models.\label{alg:TupleJoin}}
\end{algorithm}

Algorithm~\ref{alg:TupleJoin} shows the tuple nested loops join algorithm, adapted to use a large language model to evaluate join conditions. The input to Algorithm~\ref{alg:TupleJoin} are the two tables, $R_1$ and $R_2$, as well as the join condition, $j$. The join condition is formulated in natural language and expresses the condition for a match between two tuples. As the classical tuple nested loops join, Algorithm~\ref{alg:TupleJoin} iterates over all combinations of tuples from the input tables. The particularity of Algorithm~\ref{alg:TupleJoin} lies in the way the join condition is evaluated.

\begin{figure}
    \centering
    \begin{verbatim}
    Is the following true ("Yes"/"No"): [j]?
    Text 1: [t1]
    Text 2: [t2]
    Answer:
    \end{verbatim}
    \caption{Prompt template used for tuple nested loops join (instantiated by Function~\textproc{TuplePrompt} in pseudo-code).}
    \label{fig:TupleTemplate}
\end{figure}

To evaluate a join condition, Algorithm~\ref{alg:TupleJoin} performs three steps. First, it generates a prompt, instructing the language model to compare the two current tuples. Second, it invokes a language model with that prompt to execute that comparison. Finally, it interprets the text answer by the language model, adding the tuple combination to the result set if the two input tuples match.

Figure~\ref{fig:TupleTemplate} shows the template used for generating prompts. It contains several placeholders, marked by square brackets. Function~\textproc{TuplePrompt}, used in Algorithm~\ref{alg:TupleJoin}, instantiates this template by substituting placeholders with values from the input parameters. The start of the prompt template describes the task to the language model (answering the question of whether or not the following condition holds), as well as the desired output format (i.e., either a ``Yes'' or a ``No''). The instructions contain a placeholder for the join condition, \verb|[j]|, describing the conditions for a match. After that, the prompt contains the data, i.e., the two tuples to compare (represented via placeholders \verb|[t1]| and \verb|[t2]|). The prompt concludes with a request for an answer, indicating to the language model that all relevant information for the task has been conveyed.

Function~\textproc{InvokeLLM} submits prompts to a language model (e.g., GPT-4) and returns the generated answer. In principle, the generated answer could be arbitrary text. However, as the prompt specifies an expected output format, the answer should be either ``Yes'' or ``No'' in most cases. Any valid answer uses one single token. For that reason, the implementation of \textproc{InvokeLLM} configures the language model to generate at most one single output token (thereby avoiding rare but costly cases in which the language model might generate a longer text as a reply, misunderstanding the instructions).

\subsection{Cost Model}
\label{sub:tuplemodel}

The following cost model estimates (monetary) processing costs as a function of input properties. Parameters $r_1$ and $r_2$ denote the number of rows in the two input tables. Parameters $s_1$ and $s_2$ denote the (average) sizes of a tuple in the two input tables, measured in terms of the number of tokens (since Cloud providers of language models such as OpenAI charge per token processed). Also, $p$ denotes the number of tokens used for the part of the prompt that remains static across different loop iterations (i.e., all text except for the compared tuples). In some cases, generating output is more expensive than generating input. Parameter $g$ denotes the relative cost overhead of generating tokens, compared to reading tokens.

\begin{lemma}
    Comparing two input tuples incurs cost $p+s_1+s_2+g$.\label{lm:TupleComparisonCost}
\end{lemma}
\begin{proof}
    Tuple-independent parts of the prompt account for $p$ tokens read. In addition, the information about the two input tuples, i.e., $s_1+s_2$ tokens must be read. Finally, one output token (``Yes'' or ``No'') is generated in each iteration with cost $g$.
\end{proof}

Total join processing costs follow immediately.

\begin{corollary}
    Join processing costs are $r_1\cdot r_2\cdot (p+s_1+s_2+g)$.
\end{corollary}
\begin{proof}
    This follows from the cost per comparison (Lemma~\ref{lm:TupleComparisonCost}) and the number of comparisons, determined by the number $r_1\cdot r_2$ of iterations of the innermost nested loop.
\end{proof}
\section{Block Nested Loops Join}
\label{sec:BlockJoin}

This section introduces a variant of the block nested loops join, as well as an associated cost model.

\subsection{Algorithm}

\begin{algorithm}[t]
\begin{algorithmic}[1]
    \State \Comment{Perform block join between relations $R_1$ and $R_2$}
    \State \Comment{with join condition $j$ and using block sizes $b_1$ and $b_2$.}
    \Function{BlockJoin}{$R_1,R_2,j,b_1,b_2$}
    \State \Comment{Initialize result set}
    \State $R\gets\emptyset$
    \State \Comment{Partition input into batches}
    \State $\mathcal{B}_1\gets\{B_i\subseteq R_1|R_1=\dot{\cup}_iB_i,\forall i|B_i|=b_1\}$
    \State $\mathcal{B}_2\gets\{B_i\subseteq R_2|R_2=\dot{\cup}_iB_i,\forall i|B_i|=b_2\}$
    \State \Comment{Iterate over pairs of batches}
    \For{$B_1\in \mathcal{B}_1$}
    \For{$B_2\in \mathcal{B}_2$}
    \State \Comment{Create prompt for LLM}
    \State $P\gets$\Call{BlockPrompt}{$B_1,B_2,j$}
    \State \Comment{Get raw answer from LLM}
    \State $A\gets$\Call{InvokeLLM}{$P$}
    \State \Comment{Check for overflow}
    \If{$A[-1]\neq$\textbf{Finished}}
    \State \Return{\textbf{<Overflow>}}
    \EndIf
    \State \Comment{Extract result tuples}
    \State $R\gets R\cup$\Call{ExtractTuples}{$B_1,B_2,A$}
    \EndFor
    \EndFor
    \State \Comment{Return join result}
    \State \Return{$R$}
    \EndFunction
\end{algorithmic}
    \caption{Block nested loops join algorithm for semantic joins, executed via large language models.\label{alg:BlockJoin}}
\end{algorithm}

Algorithm~\ref{alg:BlockJoin} uses similar input parameters as Algorithm~\ref{alg:TupleJoin}, namely two input tables ($R_1$ and $R_2$) and a join condition $j$. In addition, Algorithm~\ref{alg:BlockJoin} uses input parameters $b_1$ and $b_2$, representing the number of tuples from the first and second table that are processed together as one batch. The choice of those parameter values is non-trivial and analyzed in the following sections.

Algorithm~\ref{alg:BlockJoin} starts by partitioning tuples from both input tables, using the specified batch sizes (the pseudo-code is slightly simplified, based on the assumption that the number of tuples in each table is a multiple of the batch sizes). Instead of iterating over pairs of tuples, the algorithm iterates over pairs of tuple batches. For each pair of batches, the algorithm uses a language model to retrieve all tuple pairs that satisfy the join condition. Instead of invoking the language model for each tuple pair, Algorithm~\ref{alg:BlockJoin} invokes the model only once for each pair of tuple batches.

\begin{figure}
    \centering
    \begin{verbatim}
    Find indexes x,y where x is the number of an entry 
    in collection 1 and y the number of an entry in 
    collection 2 such that [j] (make sure to catch 
    all pairs!)!
    Separate index pairs by semicolons.
    Write "Finished" after the last pair!
    Text Collection 1:
    1. [B1[1]]
    2. [B1[2]]
    ...
    Text Collection 2:
    1. [B2[1]]
    2. [B2[2]]
    ...
    Index pairs:
    \end{verbatim}
    \caption{Prompt template used for block nested loops join (instantiated by Function~\textproc{BlockPrompt} in pseudo-code).}
    \label{fig:BlockTemplate}
\end{figure}

Figure~\ref{fig:BlockTemplate} shows the corresponding prompt template, instantiated by Function~\textproc{BlockPrompt}. The prompt contains placeholders for the join condition, \verb|[j]|, and for the tuples in each block, denoted as \verb|[Bi[j]]| where \verb|i| is the index of the table containing the tuples and \verb|j| the index of a tuple within the current tuple batch. The template starts with instructions, directing the language model to find pairs of indexes that represent matching tuples. Each pair of matching tuples is denoted as \verb|x,y| where \verb|x| refers to the position of a tuple from the first batch and \verb|y| to the position of the tuple within the second batch. While seemingly redundant, the additional instructions \verb|make sure to catch all pairs!| are important to encourage the language model to generate a complete result. The number of matching tuple pairs may range from zero to the product of the two input batch sizes. The prompt instructs the language model to use semicolons to separate different index pairs.

The number of output tokens is limited, determined by the properties of the used language model. If reaching the limit in terms of output tokens, the answer generated by the language model becomes inconclusive. It is unclear whether the language model found all matching pairs or ran out of tokens before being able to generate complete output. For that reason, the prompt in Figure~\ref{fig:BlockTemplate} instructs the language model to mark the last matching index pair with the word ``Finished''. If the word ``Finished'' concludes the output, even when reaching the token limit, it is clear that the output contains all matching tuples (at least all matches that the language model is able to find). Finally, the prompt template contains tuples from the two input batches, each prefixed by a batch-specific index number. 

In principle, asking the language model to write complete result tuples (i.e., to copy matching input tuples) is possible as well. However, as the cost for generating output is proportional to the number of generated tokens (and, at least for some models, generating tokens is more expensive than reading tokens), generating index pairs, rather than result tuples, reduces processing fees.

Algorithm~\ref{alg:BlockJoin} sends prompts generated for the current pair of batches to the language model to retrieve an answer. First of all, Algorithm~\ref{alg:BlockJoin} checks whether a complete result (according to the capabilities of the language model) was generated. As the prompt instructs the language model to terminate output with the keyword ``Finished'', the algorithm checks the last word in the answer using the (Python-inspired) notation $A[-1]$. If the keyword is not ``Finished'', the join operator returns the flag \textbf{<Overflow>}. This means that the result is incomplete and the settings for the batch sizes, $b_1$ and $b_2$, are invalid. This can happen if initial estimates on the selectivity of the join condition, determining the number of output tokens that are generated, turn out to be erroneous. Section~\ref{sec:Adaptive} shows how to handle such cases. If no overflow occurs, the tuples associated with the index pairs are added to the result set. Function~\textproc{ExtractTuples} (the pseudo-code is omitted due to space restrictions) translates index pairs in the answer into tuple pairs.

\subsection{Cost Model}
\label{sub:CostModel}


Parameters $r_1$ and $r_2$ denote the number of rows in the first and second table respectively. Parameters $s_1$, $s_2$, and $s_3$ denote the (token) size of tuples in the two input tables and per result index pair ($s_3$), respectively. Parameter $p$ is the size of the tuple-independent parts of the prompt represented in Figure~\ref{fig:BlockTemplate} (i.e., all text except for the parts that describe the input tuples). Parameter $\sigma$ represents the selectivity of that join condition, i.e., the ratio of input tuple combinations satisfying the join condition. Finally, parameter $g$ represents the relative cost of generating tokens, relative to the cost of reading tokens. For some LLMs, the cost of reading and generating tokens is equal (i.e., $g=1$) but for some of the more recent models (e.g., GPT-4), the cost of generating tokens is higher than the cost of reading them (i.e., $g>1$). Parameters $b_1$ and $b_2$ denote the batch sizes for the first and second table (i.e., the input parameters in Algorithm~\ref{alg:BlockJoin}). Parameters related to size and selectivity (namely, parameters $s_1$, $s_2$, $s_3$, $r_1$, $r_2$, and $\sigma$) depend on data properties whereas $g$ depends on the LLM and $p$ is specified by the user. Only the values for parameters $b_1$ and $b_2$ can be chosen. 


The following lemmata and theorems calculate the number of LLM invocations, the number of tokens processed per invocation, and the cost per LLM invocation. Note that the following analysis is simplifying as it treats all parameters as continuous (e.g., $r_1/b_1$, as opposed to $\lceil r_1/b_1\rceil$, when calculating the number of batches for the first table). This facilitates the analysis in the following sections, applying differentiation to obtain optimal values for tuning parameters $b_1$ and $b_2$.


\begin{lemma}\label{lm:PromptSize}
    The number of tokens processed per LLM invocation is given by  $p+b_1\cdot s_1+b_2\cdot s_2+b_1\cdot b_2\cdot\sigma\cdot s_3$.
\end{lemma}
\begin{proof}
    Each prompt contains a batch of $b_1$ tuples from the first table with a size per tuple of $s_1$, i.e., $b_1\cdot s_1$ is the number of tokens used to represent entries from the first table. Similarly, entries from the second table consume $b_2\cdot s_2$ tokens. The expected number of join result tuples is given by $b_1\cdot b_2\cdot \sigma$ and their size by $b_1\cdot b_2\cdot \sigma\cdot s_3$. Finally taking into account tokens required for the join task description ($p$) yields the postulated size formula.
\end{proof}

\begin{lemma}\label{lm:PromptCost}
    The cost per LLM invocation is given by the formula $p+b_1\cdot s_1+b_2\cdot s_2+b_1\cdot b_2\cdot\sigma\cdot s_3\cdot g$.
\end{lemma}
\begin{proof}
    The proof is similar to the one of Lemma~\ref{lm:PromptSize}. Costs are proportional to the number of tokens, except that it distinguishes tokens read from generated tokens. The LLM only generates tokens associated with the join result. Therefore, the number of corresponding tokens ($b_1\cdot b_2\cdot\sigma\cdot s_3$) is scaled by factor $g$ to obtain associated costs.
\end{proof}

\begin{lemma}
    The number of LLM invocations for join processing is given by the formula $(r_1/b_1)\cdot(r_2/b_2)$. \label{lm:NrPrompts}
\end{lemma}
\begin{proof}
    This follows from the definition of Algorithm~\ref{alg:BlockJoin}. The LLM is called in each iteration of the inner-most loop. The outer loop iterates $r_1/b_1$ times whereas the inner loop iterates $r_2/b_2$ times.
\end{proof}

\begin{corollary}
    Total join processing costs are given by the formula $c(b_1,b_2)=(r_1/b_1)\cdot (r_2/b_2)\cdot(p+b_1\cdot s_1+b_2\cdot s_2+b_1\cdot b_2\cdot\sigma\cdot s_3\cdot g)$.
\end{corollary}
\begin{proof}
    This is a direct consequence of Lemmas~\ref{lm:PromptCost} and \ref{lm:NrPrompts}, obtained by multiplying the cost per LLM invocation with the number of LLM invocations.
\end{proof}

\section{Optimizing for Known Selectivity}
\label{sec:Optimization}

Processing fees of the block join, introduced in the previous section, depend on settings for the batch sizes (parameters $b_1$ and $b_2$). This section shows how to optimize batch sizes as a function of the input properties. The following example illustrates how processing fees depend on the batch size.

\begin{figure}[t]
    \centering
    \includegraphics{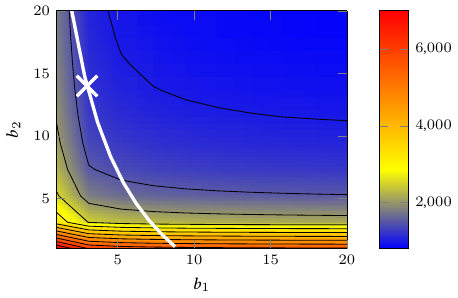}
    \caption{Illustrating join processing costs as a function of the two input batch sizes ($b_1$ and $b_2$), using $r_1=50$, $r_2=10$, $s_1=10$, $s_2=2$, $s_3=1$, $\sigma=1$, $g=1$, $p=1$. All solutions under the white curve use prompts with a size at or below 100 tokens. The white X marks the solution with minimal cost among all solutions with a prompt size of up to 100 tokens.}
    \label{fig:CostFunction}
\end{figure}


\begin{example}
Figure~\ref{fig:CostFunction} plots join cost for an example scenario. A-priori, choosing higher values for $b_1$ and $b_2$ seems preferable. However, in practice, the values of $b_1$ and $b_2$ are bounded by limits imposed by the LLM on the number of tokens read and generated per invocation. The white line in Figure~\ref{fig:CostFunction} marks value combinations for $b_1$ and $b_2$ for which the number of processed tokens reaches 100. Given a limit on processed tokens, we want to find values for $b_1$ and $b_2$ that comply with that token limit (in Figure~\ref{fig:CostFunction}, those are the points below the white line) while minimizing costs under that constraint. The white X marks the optimal solution in Figure~\ref{fig:CostFunction}.
\end{example}

\subsection{Analyzing Costs}
\label{sub:CostAnalysis}

The combined input and output size per LLM invocation is generally limited, either by a hard bound representing the maximal input and output size that a model can accept or by a (smaller) bound, representing the maximal size for which the model is deemed accurate enough. The second bound is motivated by the observation that LLMs tend to become less reliable with growing input sizes. In the following, $t$ denotes the maximal number of tokens that can be used per LLM invocation. To simplify the following formulas, $t$ does not take into account the size of the task description, $p$, which remains static over all prompts. In other words, $t$ is obtained by already subtracting $p$ from the LLM-specific size bound. To comply with the size limit, the following equation must hold.

\begin{equation}
b_1\cdot s_1+b_2\cdot s_2+b_1\cdot b_2\cdot s_3\cdot\sigma\leq t
\end{equation}

This raises the question of whether or not choosing values for $b_1$ and $b_2$ that lead to LLM invocations using less than the maximally allowed number of tokens is efficient. The following theorem shows that this is not the case.

\begin{theorem}\label{th:MaxSizeMinCost}
Maximizing the number of tokens processed per LLM invocation minimizes processing costs.
\end{theorem}
\begin{proof}
Assume that the prompt size is below the threshold, i.e., $b_1\cdot s_1+b_2\cdot s_2+b_1\cdot b_2\cdot s_3\cdot\sigma< t$. Furthtermore, without restriction of generality, assume that $b_1$ can be replaced by $b_1^*=\alpha \cdot b_1$ for an $\alpha>1$ such that $b_1^*\cdot s_1+b_2\cdot s_2+b_1^*\cdot b_2\cdot s_3\cdot\sigma\leq t$. How do total processing costs with $b_1$ ($c(b_1,b_2)$) relate to the ones with $b_1^*$ ($c(b_1^*,b_2)$)? It is $c(b_1^*,b_2)=(r_1/b_1^*)\cdot (r_2/b_2)\cdot(p+b_1^*\cdot s_1+b_2\cdot s_2+b_1^*\cdot b_2\cdot\sigma\cdot s_3\cdot g)$. This can be rewritten as $(r_1/(b_1\cdot\alpha))\cdot (r_2/b_2)\cdot(p+b_1\cdot\alpha\cdot s_1+b_2\cdot s_2+b_1\cdot\alpha\cdot b_2\cdot\sigma\cdot s_3\cdot g)$, which simplifies to $(r_1/b_1)\cdot (r_2/b_2)\cdot(p/\alpha+b_1\cdot s_1+b_2\cdot s_2/\alpha+b_1\cdot b_2\cdot\sigma\cdot s_3\cdot g)$. Since $\alpha>1$, it is $c(b_1^*,b_2)\leq (r_1/b_1)\cdot (r_2/b_2)\cdot(p+b_1\cdot s_1+b_2\cdot s_2+b_1\cdot b_2\cdot\sigma\cdot s_3\cdot g)=c(b_1,b_2)$. If replacing $b_2$ with $b_2\cdot\alpha$ with $\alpha>1$, similar reasoning shows that the cost can only decrease. Hence, increasing the number of tokens processed per LLM invocation, if possible, decreases costs.
\end{proof}

\begin{example}
Consider the cost function depicted in Figure~\ref{fig:CostFunction}. As discussed before, the white curve marks points at which the number of tokens processed per LLM invocation equals the threshold. Due to Theorem~\ref{th:MaxSizeMinCost}, values for $b_1$ and $b_2$ that minimize join processing costs while complying with token limits must be on that curve.
\end{example}


The following lemma shows that the optimal value for $b_2$ can be expressed as a function of $b_1$ (denoted as the function $b_2(b_1)$).




\begin{lemma}\label{lm:B2ofB1}
    Any solution minimizing $c(b_1,b_2)$ satisfies the equation $b_2=b_2(b_1)=(t-b_1\cdot s_1)/(s_2+b_1\cdot s_3\cdot\sigma)$.
\end{lemma}
\begin{proof}
    Due to Theorem~\ref{th:MaxSizeMinCost}, setting $b_1\cdot s_1+b_2\cdot s_2+b_1\cdot b_2\cdot s_3\cdot\sigma=t$ minimizes processing costs. This equation can be rewritten to $b_2\cdot(s_2+b_1\cdot s_3\cdot\sigma)=t-b_1\cdot s_1$. Therefore, the optimal value for $b_2$ is given as $b_2=(t-b_1\cdot s_1)/(s_2+b_1\cdot s_3\cdot\sigma)$
\end{proof}

According to Lemma~\ref{lm:B2ofB1}, substituting each occurrence of $b_2$ in the join cost function with $b_2(b_1)$ yields minimal processing costs:


\begin{align}
    &c^*(b_1):=c(b_1,b_2(b_1)) \notag \\
    =&\frac{r_1\cdot r_2}{b_1\cdot b_2(b_1)}\cdot(p+b_1\cdot s_1+b_2(b_1)\cdot s_2+b_1\cdot b_2(b_1)\cdot s_3\cdot\sigma\cdot g) \notag \\
    =&\frac{r_1}{b_1}\cdot r_2\cdot (\frac{p+b_1\cdot s_1}{b_2(b_1)}+s_2+b_1\cdot s_3\cdot\sigma\cdot g) \label{eq:CostWithB2ofB1} \notag \\
    =&\frac{r_1}{b_1}\cdot r_2\cdot (\frac{p+b_1\cdot s_1}{(t-b_1\cdot s_1)/(s_2+b_1\cdot s_3\cdot\sigma)}+s_2+b_1\cdot s_3\cdot\sigma\cdot g) \notag \\
    =&r_1\cdot r_2\cdot (\frac{(s_2/b_1+s_3\cdot\sigma)\cdot(p+b_1\cdot s_1)}{(t-b_1\cdot s_1)}+\frac{s_2}{b_1}+s_3\cdot\sigma\cdot g) \notag
\end{align}

Hence, the problem of minimizing a function with two parameters ($c(b_1,b_2)$) under constraints reduces to the problem of minimizing a function that depends on a single parameter ($c^*(b_1)$).





\subsection{Optimizing Costs}

We minimize join processing costs, i.e., $c^*(b_1)$, by a suitable choice for $b_1$. This means we are searching for minima of $c^*(b_1)$. For $b_1^*$ to be a minimum of $c^*(b_1)$, the following conditions must hold:

\begin{alignat*}{3}
\diff{c^*}{b_1}(b_1^*)=0 & \quad\quad & \diff[2]{c^*}{b_1}(b_1^*)>0
\end{alignat*}

The first-order derivative of $c^*$ is given as follows:

\begin{equation}
    \diff{c^*}{b_1}=r_1r_2(t+p)[\frac{b_1^2s_1 s_3\sigma+b_12s_1s_2-s_2t}{(t-b_1 s_1)^2b_1^2}] \label{eq:Derivative1}
\end{equation}


\begin{lemma}
    For $c^*$, $b_*=[-s_1s_2+\sqrt{s_1^2s_2^2+s_1s_2s_3\sigma t}]/(s_1s_3\sigma)$ is a critical point (i.e., the first-order derivative is zero). \label{lm:CriticalPoint}
\end{lemma}
\begin{proof}
    It is $r_1r_2(t+p)>0$ since all involved terms are positive. Similarly, it is $(t-b_1s_1)^2b_1^2>0$. Therefore, the derivative of $c^*$ reaches zero iff $b_1^2s_1 s_3\sigma+b_12s_1s_2-s_2t=0$. This is a quadratic equation in $b_1$. The roots are therefore given by $(-2s_1s_2\pm\sqrt{(2s_1s_2)^2-4(s_1s_3\sigma)(-s_2t)})/(2s_1s_3\sigma)$ which simplifies to $[-s_1s_2\pm\sqrt{s_1^2s_2^2+s_1s_2s_3\sigma t}]/(s_1s_3\sigma)$. Also, as $b_1$ represents the batch size, it must be positive. Hence, the only valid solution is $[-s_1s_2+\sqrt{s_1^2s_2^2+s_1s_2s_3\sigma t}]/(s_1s_3\sigma)$. Note that this solution is guaranteed to be positive since $s_1s_2<\sqrt{s_1^2s_2^2+s_1s_2s_3\sigma t}$.
\end{proof}

\begin{theorem}
    For $c^*$, $b_*:=[-s_1s_2+\sqrt{s_1^2s_2^2+s_1s_2s_3\sigma t}]/(s_1s_3\sigma)$ is a minimum. \label{th:GlobalMinimum}
\end{theorem}
\begin{proof}
    The theorem holds if $d^2c^*/db_1^2>0$ at $b_*$ since $b_*$ is a critical point, according to Lemma~\ref{lm:CriticalPoint}. Set $u(b_1)=b_1^2s_1 s_3\sigma+b_12s_1s_2-s_2t$ and $v(b_1)=(t-b_1 s_1)^2b_1^2$. The first-order derivative of $c^*$, $dc^*/db_1$, is $r_1r_2(t+p)u(b_1)/v(b_1)$, according to Equation~\ref{eq:Derivative1}. Due to the quotient rule, it is $d^2c^*/db_1^2=r_1r_2(t+p)[u'v-uv']/v^2$ where $u'=du/db_1$ and $v'=dv/db_1$. As outlined in the proof of Lemma~\ref{lm:CriticalPoint}, $u(b_*)=0$. Hence, at $b_*$, the second-order derivative $d^2c^*/db_1^2$ simplifies to $r_1r_2(t+p)[u'v]/v^2$. It is $u'=d/d b_1[b_1^2s_1 s_3\sigma+b_12s_1s_2-s_2t]=2b_1s_1s_3\sigma+2s_1s_2$. As all constants appearing in this equation are positive with $s_1>0$ and $s_2>0$, $u'$ is strictly positive for positive values of $b_1$. Note that $b_1s_1<t$ since the token threshold $t$ is at least equal to the number of tokens used for representing tuples from the first and second table, $b_1s_1+b_2s_2$, with $b_2s_2>0$ (since each prompt must contain non-empty input from both tables to be useful). Therefore, $v$ is strictly positive for all values of $b_1$. This implies that $d^2c^*/db_1^2$ is greater than zero at $b_*$.
\end{proof}

\begin{example}
    In the example depicted in Figure~\ref{fig:CostFunction}, we have $s_1=10$, $s_2=2$, $\sigma=s_3=1$. Therefore, it is 
    \begin{align}
        b_*=&[-s_1s_2+\sqrt{s_1^2s_2^2+s_1s_2s_3\sigma t}]/(s_1s_3\sigma) \notag \\
        =&[-10\cdot 2+\sqrt{10^2\cdot 2^2+10\cdot 2\cdot 1\cdot 1\cdot 100}]/(10\cdot 1\cdot 1) \notag \\
        =&[-20+\sqrt{2400}]/10 \approx3 \notag
    \end{align}
    This means selecting batches of three tuples from the first table is optimal (i.e., setting $b_1=b_*\approx3$). According to Lemma~\ref{lm:B2ofB1}, the optimal number of tuples per batch for the second table is determined as $b_2=(t-b_1\cdot s_1)/(s_2+b_1\cdot s_3\cdot\sigma)$ and, for $b_1=3$, it is $b_2=(100-3\cdot 10)/(2+3\cdot 1\cdot 1)=14$. Hence, setting $b_1=3$ and $b_2=14$ minimizes cost under the per-prompt token limit. In Figure~\ref{fig:CostFunction}, the white X marks that point.
\end{example}

%
\section{Adaptive Join Algorithm}
\label{sec:Adaptive}

The previous section optimizes batch sizes, assuming that the selectivity $\sigma$ of the join predicate is known. This section relaxes that assumption and shows how to deal with an unknown selectivity.


\subsection{Algorithm}

\begin{figure}
    \centering
    \begin{tikzpicture}
        \begin{groupplot}[group style={group size=1 by 2, x descriptions at=edge bottom}, xlabel={Selectivity $\sigma$}, width=6cm, height=4cm, xmin=0, xmax=1, ylabel near ticks]
            \nextgroupplot[ylabel={Batch Size}, legend entries={$b_1$, $b_2$}, legend pos=outer north east, legend columns=1, ymajorgrids, title={Optimal Batch Size}]
            \addplot+[domain=0:1, samples=50, mark size=1] {(-10*2+sqrt(10*10*2*2+10*2*1*x*100))/(10*1*x)};
            \addplot+[domain=0:1, samples=50, mark size=1] {(100-10*(-10*2+sqrt(10*10*2*2+10*2*1*x*100))/(10*1*x))/(2+x*1*(-10*2+sqrt(10*10*2*2+10*2*1*x*100))/(10*1*x))};
            \nextgroupplot[stack plots=y, area style, ylabel={\# Tokens}, legend entries={Input 1, Input 2, Output}, legend pos=outer north east, title={Optimal Allocation of Tokens}, ymin=0]
            \addplot+[domain=0:1, samples=50, pattern=horizontal lines, pattern color=black, preaction={fill,blue!50}] {10*(-10*2+sqrt(10*10*2*2+10*2*1*x*100))/(10*1*x)} \closedcycle;
            \addplot+[domain=0:1, samples=50, pattern=vertical lines, pattern color=black, preaction={fill,red!50}] {2*(100-10*(-10*2+sqrt(10*10*2*2+10*2*1*x*100))/(10*1*x))/(2+x*1*(-10*2+sqrt(10*10*2*2+10*2*1*x*100))/(10*1*x))} \closedcycle;
            \addplot+[domain=0:1, samples=50, pattern=north east lines, pattern color=black, preaction={fill,brown!50}] {100-10*(-10*2+sqrt(10*10*2*2+10*2*1*x*100))/(10*1*x)-2*(100-10*(-10*2+sqrt(10*10*2*2+10*2*1*x*100))/(10*1*x))/(2+x*1*(-10*2+sqrt(10*10*2*2+10*2*1*x*100))/(10*1*x))} \closedcycle;
        \end{groupplot}
    \end{tikzpicture}
    \caption{Impact of selectivity $\sigma$ on optimal batch sizes and token allocations for $r_1=50$, $r_2=10$, $s_1=10$, $s_2=2$, $s_3=1$, $g=1$, $p=1$, and $t=100$.}
    \label{fig:SelectivityImpact}
\end{figure}
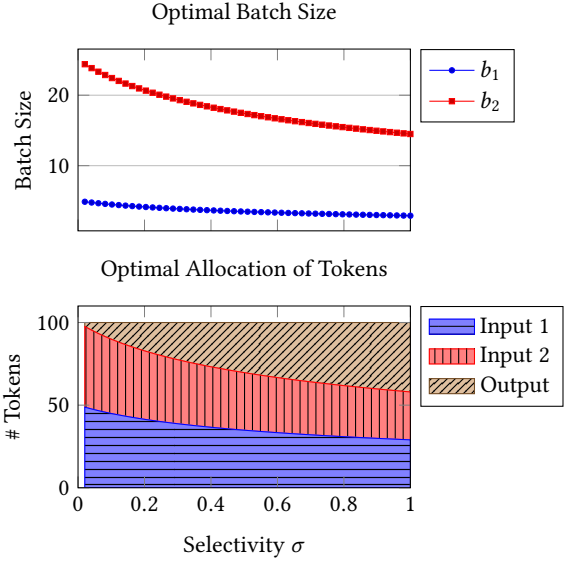

The following example illustrates the impact of selectivity.

\begin{example}
    Figure~\ref{fig:SelectivityImpact} demonstrates the impact of the join predicate selectivity. With the exception of the selectivity estimate, $\sigma$, the example uses the same settings as in Figure~\ref{fig:CostFunction}. The upper plot shows the optimal settings for the batch sizes, $b_1$ and $b_2$, as a function of selectivity (on the x-axis). The lower plot shows how tokens read or written in each LLM invocation are partitioned across tokens representing input from the first and the second table, as well as output tuples (which are written by the LLM). In the example, a higher selectivity motivates smaller batches (the analysis in the following subsection shows that this, as well as other observations from the example, generalize). Intuitively, this makes sense as the number of output tuples increases in the selectivity. Hence, keeping batch sizes constant while selectivity increases leads to an overflow, i.e., the size required for join output exceeds the maximal number of tokens. Consistent with that, the number of tokens reserved for join output increases, relative to tokens reserved for representing input, as selectivity increases.
\end{example}

The example shows that optimal choices for batch sizes depend significantly on selectivity. Traditional selectivity estimation methods, based on data statistics, \textbf{cannot be used} for join predicates in natural language. At the same time, using a selectivity estimates that deviates significantly from the actual selectivity has negative consequences.

Using a selectivity estimate that is too high is inefficient. More precisely, assuming selectivity that is too high means reserving more tokens for join output than necessary. Those tokens could be used for representing more input tuples, thereby reducing the number of iterations and, ultimately, costs. On the other hand, using a selectivity estimate that is too low is ineffective. If not reserving enough tokens for join output, the language model will be unable to generate a complete join result. In that case, the block join algorithm (Algorithm~\ref{alg:BlockJoin}) returns the \textbf{<Overflow>} flag, indicating an incomplete join result.


\begin{algorithm}[t]
\begin{algorithmic}[1]
    \State \Comment{Perform block join with between relations $R_1$ and $R_2$}
    \State \Comment{using join condition $j$ with optimistic selectivity estimate $e$.}
    \Function{AdaptiveJoin}{$R_1,R_2,j,e$}
    \State \Comment{Generate data size statistics}
    \State $stats\gets$\Call{GenerateStatistics}{$R_1,R_2,j$}
    \State \Comment{Initialize join result to overflow flag}
    \State $R\gets$\textbf{<Overflow>}
    \State \Comment{Iterate until complete join result available}
    \While{$R=$\textbf{<Overflow>}}
    \State \Comment{Calculate optimal batch sizes}
    \State $\langle b_1,b_2\rangle\gets$\Call{OptimalBatchSizes}{$stats,e$}
    \State \Comment{Try block join with those sizes}
    \State $R\gets$\Call{BlockJoin}{$R_1,R_2,j,b_1,b_2$}
    \State \Comment{Increase selectivity estimate}
    \State $e\gets e\cdot\alpha$
    \EndWhile
    \State \Comment{Return join result}
    \State \Return{$R$}
    \EndFunction
\end{algorithmic}
    \caption{Adaptive join algorithm, updating selectivity estimates as needed.\label{alg:AdaptiveJoin}}
\end{algorithm}

Fortunately, it turns out that an adaptive processing strategy, shown in Algorithm~\ref{alg:AdaptiveJoin}, achieves near-optimal performance, despite not assuming a precise selectivity estimate. Algorithm~\ref{alg:AdaptiveJoin} starts with an optimistic selectivity estimate, i.e., a selectivity estimate that is assumed to be lower than the actual selectivity. Choosing an estimate that is closer to the actual selectivity may improve performance but the effect is bounded, as shown by the analysis in the following subsection. It is, in principle, possible to start with a \textbf{higher} (i.e., pessimistic) selectivity estimate and \textbf{decrease} it to match the actual selectivity more closely. This approach is equivalent if selectivity is constant across different batches, ensuring that the selectivity observed on a sample is representative. However, in practice, selectivity differs across batches due to data skew. Hence, after lowering the selectivity estimate, meaning that less space is reserved for output in the prompt, it may be necessary to increase estimates again to avoid overflow if later batches have a higher selectivity. However, selectivity updates are undesirable as they cause overheads. When only increasing selectivity estimates, meaning that more and more space in the prompt is reserved to store output, it is never necessary to revert prior decisions and decrease estimates again to ensure that the join operator can finish (i.e., all batches are processed without overflow).

Algorithm~\ref{alg:AdaptiveJoin} calculates all relevant data statistics that appear in the formulas from Section~\ref{sec:Optimization}. For instance, this includes the average token sizes of input tuples from both tables, as well as their cardinality. After that, Algorithm~\ref{alg:AdaptiveJoin} iterates until a complete join result is generated. As a sub-function, it uses the block join algorithm, presented in Section~\ref{alg:BlockJoin}. Batch sizes are calculated, based on the current selectivity estimate. Function~\textproc{OptimalBatchSizes} encapsulates the formulas for calculating optimal batch sizes, derived in Section~\ref{sec:Optimization}. If the block join algorithm returns the \textbf{<Overflow>} flag, the selectivity estimate is increased by a factor of $\alpha$. Factor $\alpha>1$ is a tuning parameter, its impact is studied in the next subsection.

\subsection{Analysis}

The following lemmata establish properties of the optimal batch size for the first table as a function of selectivity: $b_1^*(\sigma)$.

\begin{lemma}
The optimal value for the batch size in the first table with selectivity $\sigma$, $b_1^*(\sigma)$, is anti-monotone in the selectivity $\sigma$.\label{lm:b1antimonotone}
\end{lemma}
\begin{proof}
According to Theorem~\ref{th:GlobalMinimum}, it is 
\begin{equation*}
    b_1^*(\sigma)=\frac{\sqrt{s_1^2s_2^2+s_1s_2s_3\sigma t}-s_1s_2}{s_1s_3\sigma}\,.
\end{equation*}
Multiplying numerator and denominator by $(\sqrt{s_1^2s_2^2+s_1s_2s_3\sigma t}+s_1s_2)$ yields
\begin{align*}
    b_1^*(\sigma)=&\frac{(\sqrt{s_1^2s_2^2+s_1s_2s_3\sigma t}-s_1s_2)(\sqrt{s_1^2s_2^2+s_1s_2s_3\sigma t}+s_1s_2)}{(\sqrt{s_1^2s_2^2+s_1s_2s_3\sigma t}+s_1s_2)s_1s_3\sigma} \\
    =&\frac{s_1^2s_2^2+s_1s_2s_3\sigma t-s_1^2s_2^2}{(\sqrt{s_1^2s_2^2+s_1s_2s_3\sigma t}+s_1s_2)s_1s_3\sigma} \\
    =&\frac{s_2 t}{(\sqrt{s_1^2s_2^2+s_1s_2s_3\sigma t}+s_1s_2)}\,.
\end{align*}
As the numerator does not depend on $\sigma$, while the denominator is monotone in $\sigma$, this fraction and therefore $b_1^*(\sigma)$ is anti-monotone in the selectivity $\sigma$.
\end{proof}

\begin{lemma}
    If $e\geq\sigma\geq e/\alpha$ then $b_1^*(\sigma)\leq \alpha\cdot b_1^*(e)$.\label{lm:B1UpperBound}
\end{lemma}
\begin{proof}
The following holds due to Theorem~\ref{th:GlobalMinimum} and $e/\alpha\leq \sigma$:
    \begin{align*}
        \alpha\cdot b_1^*(e)=&\alpha\cdot[-s_1s_2+\sqrt{s_1^2s_2^2+s_1s_2s_3e t}]/(s_1s_3e) \\
        =&[-s_1s_2+\sqrt{s_1^2s_2^2+s_1s_2s_3e t}]/(s_1s_3(e/\alpha)) \\
        \geq &[-s_1s_2+\sqrt{s_1^2s_2^2+s_1s_2s_3\sigma t}]/(s_1s_3\sigma) \\
        =&b_1^*(\sigma)
    \end{align*}
\end{proof}

The following lemma analyzes the product of optimal batch sizes as a function of selectivity, denoted as $b_1^*(\sigma)$ and $b_2^*(\sigma)$.

\begin{lemma}
    If $e\geq\sigma\geq e/\alpha$ then $b_1^*(\sigma)\cdot b_2^*(\sigma)\leq\alpha\cdot b_1^*(e)\cdot b_2^*(e)$.\label{lm:BatchProductBound}
\end{lemma}
\begin{proof}
    The proof uses contradiction. Assume that $b_1^*(\sigma)\cdot b_2^*(\sigma)>\alpha\cdot b_1^*(e)\cdot b_2^*(e)$. According to Lemma~\ref{lm:B1UpperBound}, it is $b_1^*(\sigma)\leq\alpha\cdot b_1^*(e)$. Therefore, $b_1^*(\sigma)\cdot b_2^*(\sigma)>\alpha\cdot b_1^*(e)\cdot b_2^*(e)$ implies $b_2^*(\sigma)>b_2^*(e)$. Due to Theorem~\ref{th:MaxSizeMinCost}, assuming selectivity $e$, the optimal values for $b_1$ and $b_2$ exploit the full number of tokens:
    \begin{equation*}
        b_1^*(e)\cdot s_1+b_2^*(e)\cdot s_2+b_1^*(e)\cdot b_2^*(e)\cdot s_3\cdot e=t
    \end{equation*}
    However, exploiting $b_2^*(\sigma)>b_2^*(e)$ and $b_1^*(\sigma)\cdot b_2^*(\sigma)>\alpha\cdot b_1^*(e)\cdot b_2^*(e)$, then anti-monotonicity of $b_1^*(\sigma)$, according to Lemma~\ref{lm:b1antimonotone}, with $e\geq\sigma\geq e/\alpha$, and, finally, $\alpha>1$, yields:
    \begin{align*}
        &b_1^*(\sigma)\cdot s_1+b_2^*(\sigma)\cdot s_2+b_1^*(\sigma)\cdot b_2^*(\sigma)\cdot s_3\cdot\sigma \\
        >&b_1^*(\sigma)\cdot s_1+b_2^*(e)\cdot s_2+\alpha\cdot b_1^*(e)\cdot b_2^*(e)\cdot s_3\cdot\sigma \\
        \geq&b_1^*(e)\cdot s_1+b_2^*(e)\cdot s_2+ b_1^*(e)\cdot b_2^*(e)\cdot s_3\cdot e=t
    \end{align*}
    This leads to a contradiction since the number of tokens used with selectivity $\sigma$ exceeds the number $t$ of available tokens.
\end{proof}

Denote by $o(e,\sigma)$ the join processing costs when optimizing for selectivity estimate $e$ while the actual selectivity is $\sigma$. 


\begin{theorem}
    If $e\geq\sigma\geq e/\alpha$ then $o(e,\sigma)\leq\alpha\cdot g\cdot o(\sigma,\sigma)$.\label{th:LowEstimateBoundedCost}
\end{theorem}
\begin{proof}
    Optimizing for an estimated selectivity of $e$, the number of model invocations for optimal batch sizes is $r_1\cdot r_2/(b_1^*(e)\cdot b_2^*(e))$, according to Lemma~\ref{lm:NrPrompts}. According to Lemma~\ref{lm:BatchProductBound}, it is $b_1^*(\sigma)\cdot b_2^*(\sigma)\leq\alpha\cdot b_1^*(e)\cdot b_2^*(e)$. Therefore, the number of model invocations when optimizing for selectivity $e$, rather than actual selectivity $\sigma$, is higher at most by factor $\alpha$: $r_1\cdot r_2/(b_1^*(e)\cdot b_2^*(e))\leq\alpha\cdot r_1\cdot r_2/(b_1^*(\sigma)\cdot b_2^*(\sigma))$. Processing costs are proportional to the number of model invocations and the cost per invocation. According to Theorem~\ref{th:MaxSizeMinCost}, any optimal choice for batch sizes leads to prompts that exploit the maximal number of tokens. The cost per prompt is therefore between $t$ (if all tokens are read) and $t\cdot g$ with $g\geq 1$ (if all tokens are written). Hence, optimizing for estimated selectivity $e$, rather than selectivity $\sigma$, can increase per-invocation costs at most by factor $g$. The postulated bound follows since the number of model invocations increases at most by factor $\alpha$ and the cost per invocation at most by factor $g$.
\end{proof}

The following theorem bounds join processing costs, assuming imprecise selectivity estimates. 


\begin{theorem}
    Given constant tuple sizes and ratios between actual and initial estimated selectivity, Algorithm~\ref{alg:AdaptiveJoin} converges to cost within factor $\alpha\cdot g$ of the optimum as the size of the input data grows.
\end{theorem}
\begin{proof}
    Assuming constant tuple sizes in both input tables, using batch sizes that are too large immediately results in an overflow (i.e., Algorithm~\ref{alg:BlockJoin} returns \textbf{<Overflow>} after a single invocation of the LLM). This means after $O(\log_\alpha(\sigma/e))$ LLM invocations, the selectivity estimate $e$ has been adapted to be at least as large as the actual selectivity $\sigma$. As $e$ and $\sigma$ are assumed constant and the maximal overhead per LLM invocation is bounded by constants as well ($t\cdot g$), the overheads due to incorrect selectivity estimates are constant as well. As the data size grows, the overheads of join processing with an estimate $e\geq\sigma$ become dominant. Also, since Algorithm~\ref{alg:AdaptiveJoin} updates estimates via multiplication by factor $\alpha$, it is $\sigma\geq e/\alpha$. According to Theorem~\ref{th:LowEstimateBoundedCost}, the cost overhead is therefore upper-bounded by factor $\alpha\cdot g$.
\end{proof}



\section{Experimental Results}
\label{sec:experiments}

The following experiments evaluate the join operators. Section~\ref{sub:setup} describes the experimental setup. Section~\ref{sub:simulation} reports on the results of simulated joins, showing how costs of different operator implementations scale as a function of the input size. Section~\ref{sub:real} reports on the results of an evaluation that uses OpenAI's GPT-4 model and compares the approaches proposed in this paper to multiple baselines.

\subsection{Experimental Setup}
\label{sub:setup}


The following experiments use a simulator as well as experiments with real LLMs. The simulator is implemented in Python~3.11. It goes beyond applying the formulas, presented in the previous sections, and simulates each single prompt instead. Unless noted otherwise, the simulation assumes a maximal context size of 8,192 tokens, a join predicate selectivity of $\sigma=0.001$, input tuple sizes of 30 tokens (i.e., $s_1=s_2=30$, this corresponds to a few sentences of text), two tokens per output tuple (i.e., $s_3=2$), and a tuple-independent prompt size of 50 tokens (i.e., $p=50$). To translate token counts into processing fees, it uses the pricing of the GPT-4 default model by OpenAI. At the time of writing, the default version charges 3 cents per 1,000 tokens read and 6 cents per 1,000 tokens generated (i.e., the relative cost of writing tokens, $g$, is two). By default, each table contains $r_1=r_2=5,000$ tuples (some experiments use larger tables, this is pointed out in the text). It is $\alpha=4$ for the adaptive join.

Beyond simulation, the experiments use OpenAI's GPT-4 model (gpt-4-0613). Join operators are implemented in Python~3.11, using OpenAI's Python client in version~1.12. GPT-4 is invoked with a per-request timeout of 20 seconds. The temperature parameter of GPT-4 is set to zero, thereby minimizing randomness in output generation. For the block join, the ``Finished'' token, marking the end of a complete join result, is used in the stopping condition for output generation (parameter ``stop''). Unless noted otherwise, GPT-4 is used with a maximal context size of 2,000 tokens. The experiments also evaluate a baseline algorithm (``embedding join''), using OpenAI's text-embedding-3-small model to calculate embedding vectors for each of the tuples in the input tables. Then, each tuple is matched to the tuple with the most similar embedding vector from the other table (based on cosine similarity). Furthermore, the experiments evaluate LOTUS~1.1.4~\cite{Patel2025}, using the default implementation of the semantic join operator. All experiments are executed on an Apple~M1 MacBook Air laptop with 16~GB of RAM, using macOS Sonoma 14.2.1.


The experiments consider three scenarios, connected to the use cases discussed in the introduction. The project code repository\footnote{\url{https://github.com/itrummer/llmjoins}} contains data generation scripts for all of the following benchmarks. The ``Emails'' scenario, loosely based on the investigation surrounding the Enron scandal, uses language models to find inconsistencies between statements made by defendants and the content of email messages, exchanged by them and their co-workers. It joins one table containing statements of the form ``[Name]: I first heard about the losses in February 2022'' with a larger table containing short emails of the form ``I first told [Name] about the losses [TimeFrame]''. Here, [Name] is one of ten common names and [TimeFrame] is a specification of a time frame that either complies, or contradicts the statement by the corresponding defendant. The scenario uses the join condition ``the two texts contradict each other.'' The second scenario (``Reviews'') is based on the IMDB movie reviews, available for instance on Kaggle\footnote{\url{https://www.kaggle.com/datasets/atulanandjha/imdb-50k-movie-reviews-test-your-bert}}. The goal is to match reviews with similar underlying sentiment (the data set comes with ground truth labels, labeling reviews as either positive or negative). As a part of the review is typically sufficient to assess the underlying sentiment, longer reviews were shortened to the first 100 tokens. The join matches the first 50 reviews with the second 50 reviews, using the join condition ``both reviews are positive or both are negative.'' The third scenario, ``Ads,'' uses language models to match ads with corresponding searches, assuming that users enter their ads and requests via free text (e.g., on a platform like Craigslist). Ads are generated from the text template ``Offering table that is [Material] and [Color]'' and searches are generated from the template ``Searching table that is [Material] and [Color]''. Here, [Material] represents a specification of the material (e.g., ``made of wood'') and [Color] a specification of the color (e.g., ``blue''). 

\subsection{Simulation Results}
\label{sub:simulation}



\begin{figure}
    \centering
    \begin{tikzpicture}
        \begin{groupplot}[group style={group size=1 by 3, vertical sep=30pt}, width=8cm, height=3cm, ymajorgrids, ylabel={Cost (USD)}, ylabel near ticks]
            \nextgroupplot[ymode=log, legend to name={leg:scalability}, legend entries={Tuple, Block-C, Block-I, Adaptive}, legend columns=4, try min ticks log=4, xlabel={Number of Tuples ($r_1$)}]
            \addplot table[x=r1, y expr=\thisrow{tuple_cost}*0.03/1000, col sep=comma] {results/scale_table_1.csv};
            \addplot table[x=r1, y expr=\thisrow{conservative_block_cost}*0.03/1000, col sep=comma] {results/scale_table_1.csv};
            \addplot table[x=r1, y expr=\thisrow{informed_block_cost}*0.03/1000, col sep=comma] {results/scale_table_1.csv};
            \addplot table[x=r1, y expr=\thisrow{adaptive_cost}*0.03/1000, col sep=comma] {results/scale_table_1.csv};

            \nextgroupplot[ymode=log, legend to name={leg:scalability}, legend entries={Tuple, Block-C, Block-I, Adaptive}, legend columns=4, try min ticks log=4, xlabel={Tokens per Tuple ($s_1$)}]
            \addplot table[x=s1, y expr=\thisrow{tuple_cost}*0.03/1000, col sep=comma] {results/scale_tuple_1.csv};
            \addplot table[x=s1, y expr=\thisrow{conservative_block_cost}*0.03/1000, col sep=comma] {results/scale_tuple_1.csv};
            \addplot table[x=s1, y expr=\thisrow{informed_block_cost}*0.03/1000, col sep=comma] {results/scale_tuple_1.csv};
            \addplot table[x=s1, y expr=\thisrow{adaptive_cost}*0.03/1000, col sep=comma] {results/scale_tuple_1.csv};

            \nextgroupplot[ymode=log, legend to name={leg:scalability}, legend entries={Tuple, Block-C, Block-I, Adaptive}, legend columns=4, try min ticks log=4, xlabel={Selectivity of Join Predicate ($\sigma$)}]
            \addplot table[x=sigma, y expr=\thisrow{tuple_cost}*0.03/1000, col sep=comma] {results/scale_output_1.csv};
            \addplot table[x=sigma, y expr=\thisrow{conservative_block_cost}*0.03/1000, col sep=comma] {results/scale_output_1.csv};
            \addplot table[x=sigma, y expr=\thisrow{informed_block_cost}*0.03/1000, col sep=comma] {results/scale_output_1.csv};
            \addplot table[x=sigma, y expr=\thisrow{adaptive_cost}*0.03/1000, col sep=comma] {results/scale_output_1.csv};
        \end{groupplot}
    \end{tikzpicture}

    \ref{leg:scalability}
    \caption{Cost of simulated joins with GPT-4.}
    \label{fig:simulation}
\end{figure}
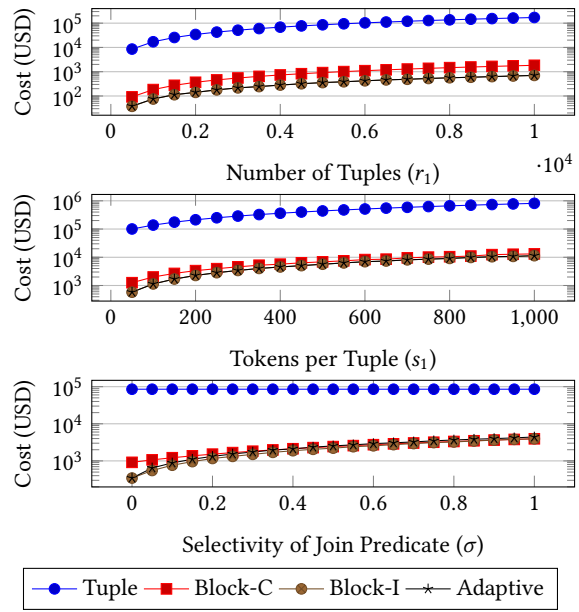

Figure~\ref{fig:simulation} compares processing costs of different join operator implementations, varying the size of the first input table, the size of the tuples ($s_1$), as well as the selectivity of the join predicate ($\sigma$). It compares the tuple join (Algorithm~\ref{alg:TupleJoin}), the block join (Algorithm~\ref{alg:BlockJoin}) when calculating batch sizes for a conservative selectivity estimate of one (which ensures enough space for result output), abbreviated as ``Block-C'', and the same algorithm when calculating batch sizes informed by the actual selectivity, abbreviated as ``Block-I''. Finally, it reports results for the adaptive join algorithm (Algorithm~\ref{alg:AdaptiveJoin}), using an optimistic selectivity estimate of $\sigma/100$ for each benchmark (i.e., initially underestimating selectivity by factor 100). The y-axis of Figure~\ref{fig:simulation} is logarithmic.

The costs of the tuple join are higher than the costs of the other join operators by several orders of magnitude. E.g., joining tables containing 10,000 and 5,000 tuples costs over 100,000 dollars when using the tuple join but less than 1,000 dollars for the Adaptive join. Among the other join operators, the block join with conservative selectivity estimates (Block-C) performs worse than the one with accurate selectivity estimates (Block-I). For instance, for an input size of 10,000 tuples, Block-C is about three times more expensive than Block-I. Block-I is difficult to implement as it requires precise selectivity estimates (which would require additional profiling mechanisms that incur additional costs). However, the adaptive algorithm performs almost identical to Block-I (e.g., cost within 0.1\% of Block-I for 10,000 input tuples) and does not require accurate selectivity estimates, making it the most practical alternative.

Increasing the number of input rows, tuple size, or join selectivity increases processing overheads for almost all operators. An exception is the tuple join for which costs do not increase when increasing join selectivity. This is expected as, unlike for the block join variants, the tuple join generates the same amount of output for matching tuple pairs as for non-matching tuple pairs. The gap between different block join variants (i.e., Block-C, Block-I, and also Adaptive) varies as a function of scenario properties. As selectivity increases, the (pessimistic) assumptions on the selectivity, underlying tuning choices made by Block-C, become accurate. Hence, the gap between block join variants shrinks as selectivity increases. 



\subsection{Benchmarks with Real LLMs}
\label{sub:real}

\begin{table}[t]
    \centering
    \caption{Benchmark statistics.}
    \begin{tabular}{llll}
    \toprule[1pt]
    \textbf{Property} & \textbf{Emails} & \textbf{Reviews} & \textbf{Ads} \\
    \midrule[1pt]
    \textbf{Tbl 1 Rows} &  100 & 50 & 16 \\
    \textbf{Tbl 2 Rows} & 10 & 50 & 16 \\
    \textbf{Tbl 1 Avg.\ Tuple Size} & 14 & 98 & 11 \\
    \textbf{Tbl 2 Avg.\ Tuple Size} & 15 & 101 & 10 \\
    \textbf{Predicate Selectivity} & 0.01 & 0.5 & 0.06 \\
    \bottomrule[1pt]
    \end{tabular}
    \label{tab:benchmarks}
\end{table}

\begin{figure}
    \centering
    \begin{tikzpicture}
        \begin{groupplot}[group style={group size=1 by 4, x descriptions at=edge bottom, vertical sep=10pt}, ybar=0pt, legend entries={Tuple-J, Block-J, Adaptive-J, Embedding-J, LOTUS}, width=8cm, height=3.25cm, xmin=0.5, xmax=3.5, ymajorgrids, every node near coord/.append style={rotate=0, anchor=south}, legend pos=outer north east, legend to name=legendRealCost, legend columns=5, xticklabels={Emails,Reviews,Ads}, xtick={1,2,3}, ylabel near ticks, ymode=log, point meta=rawy, try min ticks log=4]
            \nextgroupplot[ylabel={Fees (cents)}, ymax=20000]
            \addplot coordinates {(1,194) (2,1780) (3,43)};
            \addplot coordinates {(1,9) (2,248) (3,2)};
            \addplot coordinates {(1,7) (2,255) (3,2)};
            \addplot coordinates {(1,0.003) (2,0.02) (3,0.000007)};
            \addplot coordinates {(1,399) (2,2290) (3,96)};
            \nextgroupplot[ylabel={Tokens Read}, node near coord style={rotate=90, anchor=east}, ymin=10]
            \addplot coordinates {(1,62700) (2,588250) (3,13824)};
            \addplot coordinates {(1,2825) (2,75271) (3,466)};
            \addplot coordinates {(1,2066) (2,77192) (3,466)};
            \addplot coordinates {(1,1514) (2,9966) (3,336)};
            \addplot coordinates {(1, 127000) (2, 748300) (3, 30464)};
            \nextgroupplot[ylabel={Tokens Written}, node near coord style={rotate=90, anchor=east}, ymin=10]
            \addplot coordinates {(1,100) (2,2500) (3,256)};
            \addplot coordinates {(1,90) (2,3735) (3,80)};
            \addplot coordinates {(1,65) (2,3826) (3,80)};
            \addplot coordinates {(1,0) (2,0) (3,0)};
            \addplot coordinates {(1,3000) (2,7527) (3,768)};
            \nextgroupplot[ylabel={Time (s)}, ymax=10000]
            \addplot coordinates {(1,435) (2,1192) (3,117)};
            \addplot coordinates {(1,6) (2,160) (3,4)};
            \addplot coordinates {(1,3) (2,166) (3,5)};
            \addplot coordinates {(1,24) (2,14) (3,4)};
            \addplot coordinates {(1,13) (2,31) (3,6)};
        \end{groupplot}
    \end{tikzpicture}

    \ref{legendRealCost}
    \caption{Cost of different join operators.}
    \label{fig:RealCost}
\end{figure}
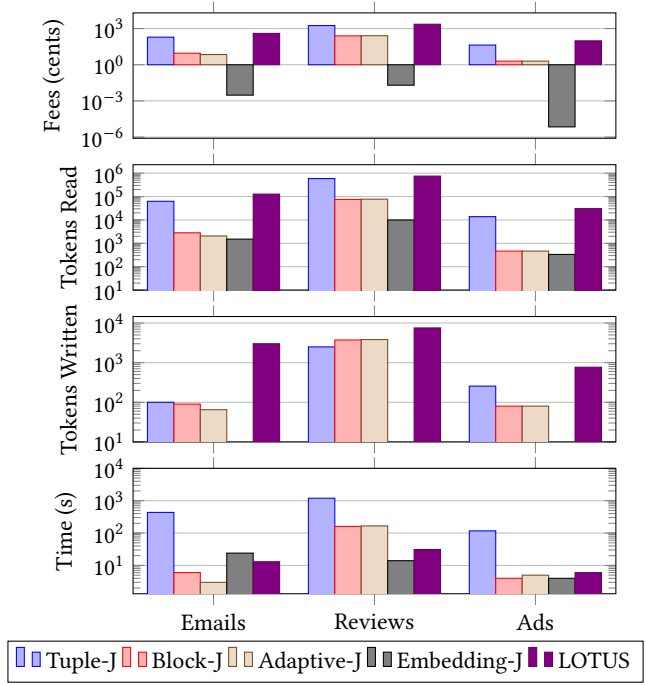


Table~\ref{tab:benchmarks} reports statistics on the benchmarks, used for the experiments in this section. Figure~\ref{fig:RealCost} reports the cost of different join operators incurred in (non-simulated) experiments with GPT-4. As in the simulation, the execution costs for the tuple join are higher than the costs for the block join variants by orders of magnitude. Due to relatively small data sizes, the gap between the adaptive join and the block nested loops join tuned using conservative selectivity estimates (i.e., $\sigma=1$) is smaller. The adaptive join is up to 30\% cheaper than the block join, while it only incurs overheads of less than 3\% in one scenario (``Reviews''). The latter scenario features the join predicate with the highest selectivity, meaning that the conservative assumptions on selectivity made by the non-adaptive block join are (almost) correct. 

The cost differences between tuple and block joins are primarily due to a large gap in terms of the number of tokens read. The number of written tokens is distributed more evenly. In the Reviews scenario, the tuple join even produces fewer tuples than the other join algorithms. This is due to the fact that the block joins require several tokens per result tuple, whereas the tuple join produces one token for each pair of tuples. As the selectivity of the join predicate is high in the ``Reviews'' scenario, the tuple join gains a slight advantage in terms of the number of generated tokens. 

Similar to processing fees, switching to the block join algorithms reduces execution time. For instance, generating a complete join result in the first scenario (``Emails'') takes 435 seconds when using the tuple join, compared to three seconds with the adaptive join algorithm. The embedding join incurs significantly lower costs than the other operators. This is due to the use of a cheaper model, generating embeddings, and to the low number of tokens read. The embedding join reads all input data only once and does not generate any output tokens.

LOTUS consumes a similar number of tokens as the tuple nested loops join algorithm. Therefore, execution costs are comparable as well and significantly higher than for the block-based join algorithms. On the other hand, LOTUS is significantly faster than the tuple nested loops algorithm. Compared to the adaptive join algorithm, LOTUS is faster in one scenario (166 versus 31 seconds), while achieving comparable execution time in another (six versus five seconds), and increasing execution time for the Emails scenario (three versus 13 seconds). Clearly, the relative performance in terms of execution time is not aligned with the relative performance in terms of the number of tokens processed. This can be explained by the fact that LOTUS parallelizes LLM invocations, whereas the implementation of the join operators proposed in this paper is sequential. While the focus of the proposed implementations is on costs, rather than run time, different blocks of input tuples could be processed in parallel as well.

\begin{figure}
    \centering
    \begin{tikzpicture}
        \begin{groupplot}[group style={group size=1 by 3, x descriptions at=edge bottom, vertical sep=10pt}, ybar=0pt, legend entries={Tuple-J, Block-J, Adaptive-J, Embedding-J, LOTUS}, width=8cm, height=3cm, xmin=0.5, xmax=3.5, ymajorgrids, every node near coord/.append style={rotate=90, anchor=east}, ymin=0, legend pos=outer north east, legend to name=legend:RealQuality, legend columns=5, xticklabels={Emails,Reviews,Ads}, xtick={1,2,3}, ylabel near ticks]
            \nextgroupplot[ylabel=Recall]
            \addplot coordinates {(1,1) (2,0.52) (3,1)};
            \addplot coordinates {(1,0.9) (2,0.46) (3,0.94)};
            \addplot coordinates {(1,1) (2,0.48) (3,0.94)};
            \addplot coordinates {(1,0) (2,0.03) (3,1)};
            \addplot coordinates {(1,1) (2,0.60) (3,1)};
            \nextgroupplot[ylabel=Precision]
            \addplot coordinates {(1,0.38) (2,0.91) (3,1)};
            \addplot coordinates {(1,0.53) (2,0.80) (3,1)};
            \addplot coordinates {(1,1) (2,0.82) (3,1)};
            \addplot coordinates {(1,0) (2,0.8) (3,1)};
            \addplot coordinates {(1,0.13) (2,0.91) (3,1)};
            \nextgroupplot[ylabel={F1 Score}]
            \addplot coordinates {(1,0.56) (2,0.66) (3,1)};
            \addplot coordinates {(1,0.67) (2,0.59) (3,0.97)};
            \addplot coordinates {(1,1) (2,0.60) (3,0.97)};
            \addplot coordinates {(1,0) (2,0.06) (3,1)};
            \addplot coordinates {(1,0.23) (2,0.72) (3,1)};
        \end{groupplot}
    \end{tikzpicture}

    \ref{legend:RealQuality}
    \caption{Output quality of different join operators.}
    \label{fig:RealQuality}
\end{figure}
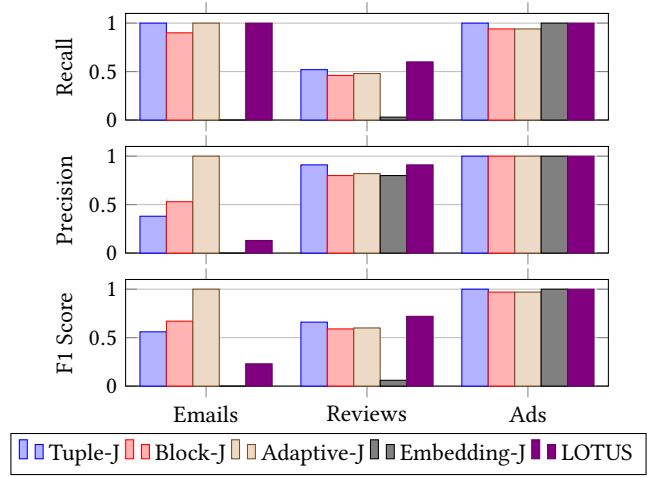

Figure~\ref{fig:RealQuality} reports on the accuracy of different join operators. Specifically, it reports recall, precision, and the F1 score, measured by comparing the result tuples generated by different join operators to the ground truth result. In two of the three scenarios, using block joins, rather than the tuple join, leads to a slight degradation of F1 scores. However, in the first scenario, using the adaptive join over the tuple join almost doubles the F1 score. It seems that GPT-4 is able to identify pairs of contradicting statements better when seeing a larger sample of all available statements. This shows that, despite reducing costs and time by orders of magnitude, using block joins over the simple tuple join does not degrade result quality in general.

The embedding join has an F1 score of zero (with \emph{both}, precision and recall, evaluating to zero) for the Email benchmark and an F1 score of 0.06 for the reviews. On the other hand, it has a perfect F1 score of one for Ads. This can be explained by the properties of the join predicates. For Ads, the goal is to find matching ads and searches. Here, having similar embedding vectors is indeed a good indicator for whether or not two tuples satisfy the join condition. For the Emails benchmark, on the other hand, the goal is to find contradicting statements. Such statements likely have dissimilar embedding vectors. 

LOTUS achieves an optimal F1 Score in two out of the three scenarios. Interestingly, the F1 Score for the Emails scenario is significantly below the block-based join algorithms. This correlates with the result quality of the embedding-based join algorithm. LOTUS uses embeddings to speed up joins as well. Hence, this scenario, aimed at finding contradicting text, appears to be hard for embedding-based methods in general.

\section{Related Work}
\label{sec:Related}

This work relates most to several recently proposed systems for semantic query processing~\cite{Patel2025, Madden2024, Liu2025b, Urban2024, Jo2024a, Jo2023a}, enabling users to formulate queries that go beyond the capabilities of pure SQL. Many of those systems support variants of semantic join operators. For instance, Section~\ref{sec:experiments} compares the proposed join operator implementations to the one used in the LOTUS system. The block-based join operator implementations described in this paper could be integrated into those systems as well. By its focus on implementing semantic versions of relational operators efficiently, this work relates to another recent paper~\cite{Shao2025}. In contrast to joins, the aforementioned paper focuses on efficient implementations of semantic sort operators.

Join algorithms have been the focus of intensive research in the database community for many decades~\cite{Shapiro1986}. The join operators proposed in this paper are variants of nested loop joins, the most popular join operator for theta joins in general. However, the focus on language models implies several unique constraints, influencing not only the operator implementations but also the associated cost models and, therefore, the optimal settings for parameters such as tuple batch sizes. First, using simple, traditional cost models (based on the number of pages read and written), nested loop join variants require only one single output buffer page, independently of the join result size. This means that join selectivity does not influence optimal batch sizes for the input tables. Instead, for language models, the number of output tuples influences the number of tokens available for reading input. Second, traditional block nested loop join variants assume that we can load additional data into an input buffer while maintaining the content of other input buffers at no additional costs. Instead, language models incur costs for reading all relevant input tokens repeatedly, independently of whether the content changed, compared to the last invocation, or not. Therefore, maximizing the size of one input buffer while minimizing the size of the other, a strategy that works best for block nested loops join in a traditional setting, does not maximize performance when executing joins via language models (e.g., this becomes apparent in Figure~\ref{fig:CostFunction}).

This work connects to prior work that exploits language models for data management tasks~\cite{Trummer2022, Narayan2022, Arora, Kayali2023, Saeed2023, Trummer2022b, Chen2023, Thorne2021}. In particular, it connects to prior work leveraging language models for join processing~\cite{Suri2021}. However, prior work focuses on similarity-based joins (i.e., items match if they are more similar) and proposes a task-specific training phase. In contrast to that, the approach presented in this paper supports generic theta joins. The join condition is specified in natural language and may, in fact, connect tuples because they are dissimilar (e.g., matching tuples that represent contradicting statements, a scenario evaluated in Section~\ref{sec:experiments}). Also, unlike prior work requiring a task-specific training phase, the approaches presented in this paper focus on a zero-shot scenario, avoiding the need for task-specific training labels. Different from other recent work~\cite{Saeed2023}, the approaches presented here assume that input data needs to be fed as input to the language model (rather than extracting information contained in the learned weights of the model).

As pointed out in a recent vision paper~\cite{Parameswaran2023}, implementing relational operators with language models connects to prior work leveraging crowdsourcing for data processing~\cite{Marcus2011, Franklin2011, Parameswaran2012, Parameswaran2013}. In particular, it connects to prior work leveraging human crowd workers for joins and related matching tasks~\cite{Whang2013, Wang2013e, Marcus2011a, Marcus2012a, Demartini2012}. However, crowdsourcing adds specific challenges (e.g., the need to aggregate diverging answers from different crowd workers) whereas it removes others (e.g., hard bounds on the combined input and output size for each task), thereby motivating different algorithmic design decisions. Broadly, this work connects to prior approaches, adapting join algorithms to new processing contexts, e.g., multi-core architectures~\cite{Albutiu2012a, Balkesen2013}, GPUs~\cite{Kaldewey2012, Yuan2013a}, and FPGAs~\cite{Halstead2013}. The approaches presented in this paper target a different platform (namely: language models) with unique properties. 

The work presented here also differs from recent work, leveraging machine learning to speed up traditional, relational processing~\cite{Sabek2021}. Instead, this paper aims to expand the scope of relational processing via language models. The adaptive join algorithm connects to a rich body of work on adaptive query processing~\cite{Avnur2000, Borovica-Gajic2018, Kaftan2018, Trummer2021c}. However, the adaptive algorithm presented here aims at solving specific challenges that arise in the context of language models, in particular, the need to balance the input size with the expected output size.


\section{Conclusion}
\label{sec:Conclusion}

This paper introduces, analyzes, and evaluates multiple variants of a novel implementation of the semantic join operator. Different from implementations used in current semantic query processing engines, this implementation integrates batches of rows into each prompt, thereby reducing the number of LLM invocations. This leads to significant performance advantages compared to prior operator implementations.



\balance
\bibliographystyle{ACM-Reference-Format}
\bibliography{library,library2}

\end{document}